 \documentclass[letterpaper, 10 pt, journal, twoside]{ieeetran}

\usepackage{amsmath}

\usepackage{amsthm}
\usepackage{amsfonts}
\usepackage{graphicx}
\usepackage{cite}
\usepackage{nicefrac}
\usepackage{eufrak}
\usepackage{mathrsfs}
\usepackage{upgreek}
\usepackage{xcolor}

\newtheorem{defn}{Definition}
\newtheorem{thm}{Theorem}

\newtheorem{prop}{Proposition}
\newtheorem{cor}{Corollary}

\newtheorem{rem}{Remark}

\newcommand{\ee}{\varepsilon}
\newcommand{\R}{\mathbb{R}}

\newcommand{\N}{\mathcal{N}}
\newcommand{\LL}{\mathcal{L}}
\newcommand{\B}{\mathcal{B}}

\newcommand{\U}{\mathcal{U}}
\newcommand{\A}{\mathscr{A}}

\newcommand{\vol}{\mathrm{vol}}
\newcommand{\X}{\mathfrak{X}}
\begin{document}
\title{On the Sample Complexity of Data-Driven Inference of the $\mathcal{L}_2$-gain} 
\author{Miel Sharf
\thanks{M. Sharf is with the Faculty of Aerospace Engineering, Israel Institute of Technology, Haifa, Israel.
    {\tt\small msharf@tx.technion.ac.il, mielsharf@gmail.com}.}
}

\maketitle
\thispagestyle{empty} 
\begin{abstract}
Lately, data-driven control has become a widespread area of research. A few recent big-data based approaches for data-driven control of nonlinear systems try to use classical input-output techniques to design controllers for systems for which only a finite number of (input-output) samples are known. These methods focus on using the given data to compute bounds on the $\mathcal{L}_2$-gain or on the shortage of passivity from finite input-output data, allowing for the application of the small gain theorem or the feedback theorem for passive systems. One question regarding these methods asks about their sample complexity, namely how many input-output samples are needed to get an approximation of the operator norm or of the shortage of passivity. We show that the number of samples needed to estimate the operator norm of a system is roughly the same as the number of samples required to approximate the system in the operator norm.
\end{abstract}


\section{Introduction}
In recent years, technological advancements have allowed to store large amounts of data. These advancements incited new problems related to analysis of the data, inference from the data, and mining said data, grouped together under the field of ``big data" \cite{Chen2014}. In engineering applications, data can be gathered from experiments or numerical models. Different methods for using big data in controller design have been offered, as summarized by \cite{Hou2013,Tabuada2017} and the references therein, which focus on the case where the system is governed by a differential equation of a known form. However, in many systems for which the big data approaches are needed, data is gathered from large, complex, and uncertain models, meaning that this assumption cannot be justified. One approach to this problem revolves around identifying an approximate model for the system, and using it for controller synthesis. These methods include system-level synthesis \cite{Dean2017,Recht2019}, which usually requires the system to be linear and time-invariant, as well as system identification techniques. Some of these techniques are tailored for linear and time invariant systems, i.e. frequency-domain methods \cite{Ljung1999,McKelvey1996}, but assuming a system is linear and time invariant can be far from true, especially for systems which require the application of big-data techniques. Other system identification methods are tailored for nonlinear systems, see \cite{Sjoberg1995,Schoukens2019} and reference therein, but can be very complex and require large amounts of data.

Recently, data-driven model-free control was proposed to overcome this problem. The idea of data-driven model-free control is to use data to solve the controller synthesis problem directly, without identifying an approximate model first. For linear and time-invariant (LTI) systems, one prominent method of data-driven model-free control comes from Willems' lemma, which characterizes all possible trajectories from a persistently exciting input \cite{Willems2005}. Different methods use this lemma to characterize all stabilizing controllers and solve the linear-quadratic regulation problem \cite{DePersis2019}, do model-predictive control \cite{Coulson2018}, and design the control law to remain within a given set \cite{Bisoffi2019}. Another idea, which also works for non-LTI systems, is to use the data sampled from the system to give an overestimate on its $\LL_2$-gain or its shortage of passivity, which is then used for controller design together with the small gain theorem or the feedback theorem for passive systems \cite{Khalil2002}. In this direction, some methods assume the system is LTI, and use the system matrix to write the passivity index or $\LL_2$-gain as a Rayleigh quotient, which is then computed using Willems' lemma, gradient descent or Gaussian processes \cite{Romer2019,Romer2017b,Romer2019b,Tanemura2019}. Other methods were proposed for nonlinear systems, based on sampling inputs which are $\ee$-nets \cite{Montenbruck2016} or on convex cone theory \cite{Romer2017a}.

One important question regarding these methods is the sample complexity, namely the amount of data needed to compute the $\LL_2$-gain or the shortage of passivity. Willems' lemma shows that a single input-output trajectory can be enough to compute these quantities if the system is known to be linear and time-invariant \cite{Romer2019}. However, the sample complexity of computing the $\LL_2$-gain or the shortage of passivity for nonlinear systems is not known. In this paper, we show that the sample complexity of computing an $\ee$-close estimate of the $\LL_2$-gain for general, nonlinear systems is roughly as large as the sample complexity of computing an $\ee$-close approximation of the system in the operator norm. We do so by showing that these sample complexities can be understood using a geometric notion called the cover index, namely the number of balls of radius $\ee$ needed to cover a certain set. We also give an estimate on this cover index, which translates to an estimate of the amount of data needed to give an $\ee$-close estimate of the $\LL_2$-gain.

\paragraph*{Notation}
We fix a time interval $I$, and denote the collection of square-integrable signals from $I$ to $\R$ as $\mathcal{L}_2(I)$, or as $\LL_2$ when $I$ is clear. The $\mathcal{L}_2$ norm of a signal $u \in \LL_2(I)$ is denoted as $\|u\| = \sqrt{\int_I |u(s)|^2{\rm d}s}$. Moreover, given a subset $\U$ of $\mathcal{L}_2$, we denote the collection of all $L$-Lipschitz operators $H:\U\to \LL_2$ by ${\rm Lip}_L(\U,\LL_2)$. 

\section{Background and Problem Formulation}
We consider a system $H$ as an operator, which takes an input signal $u\in \LL_2(I)$ and returns an output $y\in \LL_2(I)$. We consider a set of $\U \subseteq \LL_2(I)$ of admissible inputs, which will usually contain signals with a bound on their size, their frequency, and/or their energy.  We think of the system $H$ as defined only on $\U$, namely $H:\U\to \LL_2(I)$. The $\LL_2$-gain of the system $H$ is defined as the operator norm of the map $H:\U \to\LL_2(I)$, i.e. as $\|H\| = \max_{0\neq u \in \U} \frac{\|H(u)\|}{\|u\|}$. 
For simplicity, we follow \cite{Montenbruck2016} and assume that $\U$ is a closed bounded set, and that either $0 \not\in \U$, or that the maximum defining the $\LL_2$-gain is achieved on the set $\U\cap\{u\in \LL_2:\ \|u\| \ge \mu\}$ for some known $\mu > 0$. In the latter case, we can replace $\U$ by $\U\cap\{u\in \LL_2:\ \|u\| \ge \mu\}$, and may assume without loss of generality that $0\not\in \U$. This assumption is not too restrictive, as the system at hand needs to be controlled, and cannot be left evolving autonomously. If one knows an overestimate of the $\LL_2$-gain of the system $H$, one can apply the small gain theorem in order to synthesize different controllers for $H$ \cite{Khalil2002}.
In our problem, we assume that we are given knowledge of the operator $H$ when restricted to some finite set $\U^\prime = \{u_1,\cdots,u_N\}$, i.e. of the outputs $y_i = H(u_i)$ for $i=1,\cdots,N$ defining the operator $H|_{\U^\prime}:\U^\prime \to \LL_2$. These outputs can come either from experiments or from a detailed simulation of the system, which might be too complex for designing controllers. We study algorithms which utilize knowledge of the set $\U^\prime$ and of the restricted operator $H|_{\U^\prime}$ to give a bound on the $\LL_2$-gain of the unrestricted operator $H$. We also allow the algorithm to choose the inputs $u_1,\cdots,u_N$. To do so, we first define the notion of a sampling algorithm.
\begin{defn}
Let $\U$ be any subset of $\LL_2$, and let $H \in {\rm Lip}_L(\U,\LL_2)$ be an $L$-Lipschitz operator $H:\U \to \LL_2$. An $N$-sample sampling algorithm $\mathcal{S}_N$ is an algorithm choosing inputs $u_1,\cdots,u_N \in  \U$ and sampling the outputs $y_i = H(u_i)$ for $i=1,\cdots,N$. Each input $u_i$ can only depend on previous data, i.e. on $\{u_k,y_k:\ k\le i-1\}$.
\end{defn}
Out of the sampled data, we wish to construct an overestimate of the $\LL_2$-gain of the system $H$, which is close to the true $\LL_2$-gain of the system. We also consider a more complex problem, in which we try to find an approximation of the system $H$ in the operator norm:
\begin{defn}
Let $\U$ be any subset of $\LL_2$, and let $\ee,L>0$ be any positive numbers.
\begin{itemize}
\item[i)] An $N$-sample $\LL_2$-gain estimation algorithm $\A$ on ${\rm Lip}_L(\U,\LL_2)$ is comprised of an $N$-sample sampling algorithm and a map $f:(u_1,y_1,\cdots,u_N,y_N) \mapsto \gamma \in \mathbb{R}$. We say that $\A$ provides $\ee$-close overestimation of the $\LL_2$-gain if for any $H \in {\rm Lip}_L(\U,\LL_2)$, the number $\gamma = f(u_1,y_1,\cdots,u_N,y_N)$ satisfies $\|H\| \le \gamma \le \|H\| + \ee$, where $(u_1,y_1,\cdots,u_N,y_N)$ are the data gathered by the sampling algorithm.
\item[ii)] An $N$-sample norm-approximation algorithm $\A$ on ${\rm Lip}_L(\U,\LL_2)$ is comprised of an $N$-sample sampling algorithm and a map $f:(u_1,y_1,\cdots,u_N,y_N) \mapsto H_1\in {\rm Lip}_L(\U,\LL_2)$. We say that $\A$ provides $\ee$-close operator approximation if for any $H \in {\rm Lip}_L(\U,\LL_2)$, the output $H_1 = f(u_1,y_1,\cdots,u_N,y_N)$ satisfies $\|H-H_1\| \le \ee$, where $(u_1,y_1,\cdots,u_N,y_N)$ are the data gathered by the sampling algorithm.
\end{itemize}
\end{defn}
After defining what algorithms solve our problem, we can define the corresponding sample complexity as the minimum number of samples needed to solve the problem, namely:
\begin{defn}
Let $\U$ be any subset of $\LL_2$, and let $\ee,L>0$ be any positive numbers.
\begin{itemize}
\item[i)] The smallest number $N$ such there exists an $N$-sample $\LL_2$-gain estimation algorithm $\A$ providing $\ee$-close overestimation of the $\LL_2$-gain for any $H\in {\rm Lip}_L(\U,\LL_2)$ will be denoted as $\N_{\LL_2}(\U,\ee,L)$. If no such $N$ exists, we take $\N_{\LL_2}(\U,\ee,L) = \infty$.
\item[ii)] The smallest number $N$ such there exists an $N$-sample norm-approximation algorithm $\A$ providing $\ee$-close operator approximation for any $H\in {\rm Lip}_L(\U,\LL_2)$ will be denoted as $\N_{\rm op}(\U,\ee,L)$. If no such $N$ exists, we take $\N_{\rm op}(\U,\ee,L) = \infty$.
\end{itemize}
\end{defn}

\begin{rem}
We should note that if $\ee$ is on the same scale as $L$, the problem of giving an $\ee$-close overestimate of the $\LL_2$-gain becomes easy. Indeed, assume for a second that the set $\U$ is symmetric, in the sense that $u\in \U$ if and only if $-u\in \U$. Take $u\in {\rm arg}\min_{v\in \U} \|v\|$, and sample both $H(u)$ and $H(-u)$. Then for any other point $u^\prime \in \U$, either $\langle u,u^\prime\rangle \ge 0$ or that $\langle -u,u^\prime\rangle \ge 0$, so by the cosine theorem we conclude that either $\|u-u^\prime\| \le \sqrt{\|u\|^2 + \|u^\prime\|^2}$ or $\|u+u^\prime\| \le \sqrt{\|u\|^2 + \|u^\prime\|^2}$, and in particular either $\frac{\|u-u^\prime\|}{\|u^\prime\|} \le \sqrt{2}$ or $\frac{\|u+u^\prime\|}{\|u^\prime\|} \le \sqrt{2}$. In the former case, we get:
\small
\begin{align*}
\frac{\|H(u^\prime)\|}{\|u^\prime\|} \le \frac{\|H(u)\|+L\|u-u^\prime\|}{\|u^\prime\|}\le \frac{\|H(u)\|}{\|u\|}+\sqrt{2}L,
\end{align*}\normalsize
and in the latter case, we get:
\small
\begin{align*}
\frac{\|H(u^\prime)\|}{\|u^\prime\|} \le \frac{\|H(-u)\|+L\|u+u^\prime\|}{\|-u^\prime\|}\le \frac{\|H(-u)\|}{\|-u\|}+\sqrt{2}L
\end{align*}\normalsize
Thus, we can take $\gamma = \max\left\{\frac{\|H(u)\|}{\|u\|},\frac{\|H(-u)\|}{\|-u\|}\right\} + \sqrt{2}L$. It's clear that $\gamma \le \|H\|+\sqrt{2}L$, as $\|H\| \ge \max\left\{\frac{\|H(u)\|}{\|u\|},\frac{\|H(-u)\|}{\|-u\|}\right\}$, and we showed that $\gamma \ge \|H\|$. Similarly, the problem is still relatively easy so long that $\ee \approx L$. For this reason, we are interested in the case $\ee \ll L$.
\end{rem}

\begin{rem}
Proving an upper bound on $\N_{\LL_2}(\U,\ee,L)$ or $\N_{\rm op}(\U,\ee,L)$ is relatively simple, as one only needs to present some $N$-sample algorithm solving the corresponding problem. However, providing a lower bound on $\N_{\LL_2}(\U,\ee,L)$ or $\N_{\rm op}(\U,\ee,L)$ can be harder, as one needs to show that no $N$-sample algorithm can solve the corresponding problem, no matter what actions it takes. Moreover, it's clear that $\N_{\LL_2}(\U,\ee,L)$ and $\N_{\rm op}(\U,\ee,L)$ both increase as $\ee$ decreases, as if $\ee_1<\ee_2$, any $\ee_1$-close approximation is also a $\ee_2$-close approximation, and any $\ee_1$-close estimate of the $\LL_2$-gain is also a $\ee_2$-close estimate of the $\LL_2$-gain.
\end{rem}

\begin{rem} \label{rem.NaturalHarder}
We note that norm-approximation is at least as hard as $\LL_2$-gain estimation, at least asymptotically in $\ee$. Precisely, we show that $\N_{\LL_2}(\U,\ee,L)\le\N_{\rm op}(\U,\ee/2,L)$. Indeed, if we have an $N$-sample norm-approximation algorithm $\A$ providing $\ee/2$-close operator approximation, for any $H\in {\rm Lip}_L(\U,\LL_2)$ the algorithm outputs some $H_1$ such that $\|H-H_1\| \le \ee/2$. We define an algorithm $\mathscr{B}$ which runs $\A$, and then outputs $\gamma = \|H_1\|+\ee/2$ as an estimate to the $\LL_2$-gain of $H$. The triangle inequality for $\|H-H_1\|\le \ee/2$ shows that $\gamma \ge \|H\|$ and that $\gamma \le \|H\|+\ee$. Thus $\mathscr{B}$ is an $N$-sample $\mathcal{L}_2$-gain estimation algorithm providing $\ee$-close overestimation of the $\LL_2$-gain for any $H\in {\rm Lip}_L(\U,\LL_2)$, proving the claim.
\end{rem}

\section{The Projective Distance}
The basic question that determines the number of samples needed to solve a learning problem regards \emph{generalization}. Namely, how much can we learn from one (or a few) measurements? In our case, a sample $(u,y=H(u))$ can obviously be used to give a lower bound to the global $\mathcal{L}_2$-gain by $\|y\|/\|u\|$. However, it can also be used to give an upper bound on the $\mathcal{L}_2$-gain, at least near $u$. Namely, it is shown in \cite{Montenbruck2016} that if the operator $H$ is $L$-Lipschitz, then for every $u^\prime$ such that $\|u^\prime - u\| \le \delta$,
\small
\begin{align}
\frac{\|H(u^\prime)\|}{\|u^\prime\|} \le \frac{\|u\|}{\|u\|-\delta}\frac{L\|u-u^\prime\| + \|H(u)\|}{\|u\|}.
\end{align}\normalsize
Thus, a sample $(u,y=H(u))$ essentially gives an estimate for each point in the set $\B_{\|\cdot\|,\delta\|u\|,u}$, which is the closed norm ball around $u$ of radius $\delta\|u\|$, meaning that the number of samples needed to give an $\ee$-approximation of the $\mathcal{L}_2$-gain of the unknown operator $H$ can be bounded by the number of norm balls $\B_{\|\cdot\|,\delta\|u\|,u}$ needed to cover the set $\U$. Estimates on the number of copies of a shape $B$ needed to cover a shape $A$ inside a finite-dimensional vector space have been extensively studied over the last few decades, due to their connection to learning theory and Gaussian processes \cite{ShalevShwartz2014,Dudley1967}. However, the case in which different shapes (or shapes of different sizes) are used for the covering is significantly less explored. Our goal in this section is to present an alternative definition of distance, $d(x,y)$ for any two points $x,y\in\LL_2$ (or more exactly, for $x,y\in\LL_2\setminus\{0\}$) for which the sets $\B_{\|\cdot\|,\delta\|u\|,u}$ are (approximately) balls of uniform sizes for the new definition of distance. One property that will become important is invariance under scalar multiplication, i.e. that $d(x,y) = d(\alpha x,\alpha y)$ for any constant $\alpha \neq 0$ and any $x,y\in \LL_2$. To do so, we define a distance $d$ on $\LL_2\setminus\{0\}$:
\begin{defn}
Let $\X$ be any (possibly infinite-dimensional) normed space. Define the \emph{projective distance} on $\X\setminus\{0\}$ by:
\begin{align*}
d(x_1,x_2) = \frac{\|x_1-x_2\|}{\max\{\|x_1\|,\|x_2\|\}}
\end{align*}
\end{defn}

We first ask what properties does the projective distance $d$ satisfy:
\begin{defn}[\hspace{-0.2pt}\cite{Wilansky2008}]\label{def.SemiMetric}
Let $X$ be a set. A map $d:X\times X \to [0,\infty)$ is called a semi-metric if it satisfies the following properties:
\begin{itemize}
\item[i)] For all $x_1,x_2 \in X$, $d(x_1,x_2) \ge 0$.
\item[ii)] For all $x_1,x_2\in X$, $d(x_1,x_2) = 0$ if and only if $x_1=x_2$.
\item[ii)] For all $x_1,x_2 \in X$, $d(x_1,x_2) = d(x_2,x_1)$.
\end{itemize}
The pair $(X,d)$ is called a semi-metric space.
\end{defn}

\begin{prop}
The projective distance $d:(\X\setminus\{0\})\times(\X\setminus\{0\}) \to [0,\infty)$ is a semi-metric. Moreover, $d$ is invariant under scalar multiplication.
\end{prop}
\begin{proof}
Follows immediately from the properties of the norm $\|\cdot\|$ and the definition of $d$.
\end{proof}

The most common definition of distance on an arbitrary set is known as a metric \cite{Wilansky2008}. A metric is any semi-metric, as in Definition \ref{def.SemiMetric}, which also satisfies the triangle inequality, i.e. for all $x,y,z\in X$, the following inequality holds:
\begin{align*}
d(x,z) \le d(x,y) + d(y,z).
\end{align*}
Unfortunately, one can show that whenever $\dim \X \ge 2$, the projective distance is not a metric. However, this fact will not hinder our use of the projective distance. Indeed, one important property of the projective distance that will be used repeatedly is that the ball of radius $\ee\ll 1$ around $x\in \X\setminus\{0\}$ with respect to the distance $d$, is roughly equal to the normed ball of radius $\ee\|x\|$ around $x$. Namely:
\begin{prop}\label{thm.EquivalentOfDistances}
Let $\ee < 1$ be positive, and let $x,y \in \X\setminus\{0\}$.
\begin{itemize}
\item[i)] If $d(x,y)\le \ee$ then $\|x-y\| \le \frac{\ee}{1-\ee}\|x\|$ and $\|x-y\| \le \frac{\ee}{1-\ee}\|y\|$
\item[ii)] If $\|x-y\| \le \ee \|x\|$ then $d(x,y) \le \ee$.
\end{itemize}
\end{prop}
The proposition is illustrated in Fig. \ref{fig.Visualization}. In particular, if $\ee \ll 1$, $\frac{\ee}{1-\ee} \approx \ee$, so the ball around $x$ of radius $\ee$ with respect to the projective distance is roughly equal to the normed ball of radius $\ee\|x\|$.
 We now prove the theorem.
\begin{figure} [!t] 
    \centering
    \includegraphics[scale=0.32]{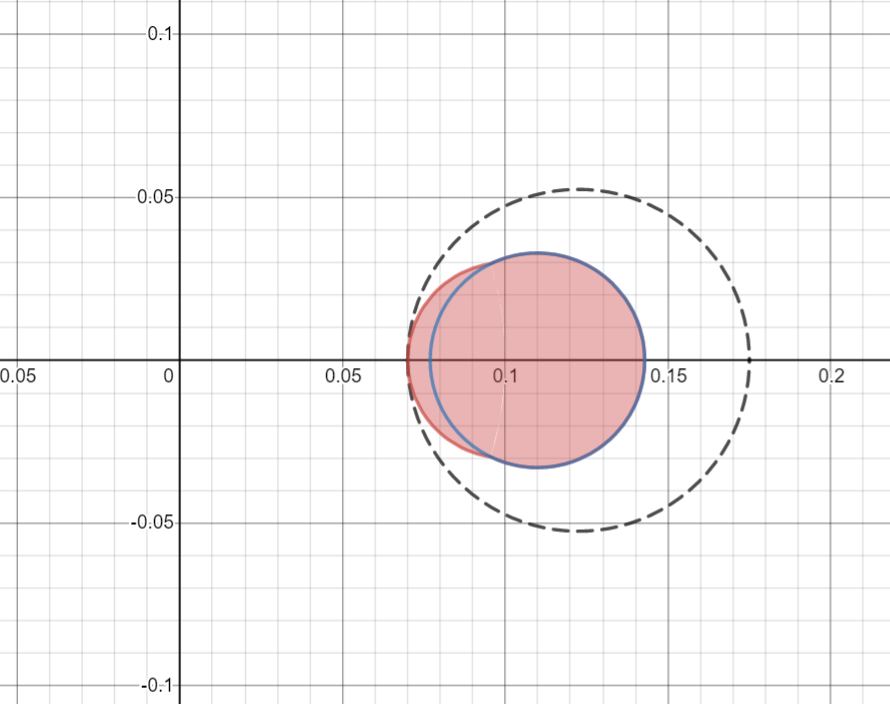}
    \caption{Visualization of Proposition \ref{thm.EquivalentOfDistances} for $\X = \R^2$. The red set is $d(x,y)\le \ee = 0.3$, where $y=(0.1,0)$, the blue line is the boundary of the set $\{\|x-y\|\le \ee \|x\|\}$ and the dashed black line is the boundary of the set $\{\|x-y\|\le \frac{\ee}{1-\ee} \|x\|\}$. The red set contains the interior of the blue line, and is contained within the dashed black line.}
    \label{fig.Visualization}
	\vspace{-16pt}
\end{figure}

\begin{proof}
We start with i). By definition, $\ee \max\{\|x\|,\|y\|\} \ge \|x-y\|$. If $\|y\| \le \|x\|$, we are done , as $\ee < \frac{\ee}{1-\ee}$. Otherwise,
\begin{align} \label{eq,WrongD}
\|x-y\| \le \ee \|y\|,
\end{align}
which gives $\|x\| \ge (1-\ee)\|y\|$ by the triangle inequality. Recalling \eqref{eq,WrongD}, we get $\|x-y\| \le \frac{\ee}{1-\ee} \|x\|$. Thus $\|x-y\| \le \frac{\ee}{1-\ee} \|x\|$ holds in both cases. Reversing the roles of $x,y$ gives $\|x-y\|\le \frac{\ee}{1-\ee}\|y\|$. As for ii), note that $\|x\| \le \max\{\|x\|,\|y\|\}$, so $\|x-y\| \le \ee\|x\| \le \ee\max\{\|x\|,\|y\|\}$, and $d(x,y)\le\ee$.
\end{proof}

\begin{cor} \label{cor.EquivDistances2}
Let $x,y \in \X\setminus\{0\}$ be arbitrary, and let $\eta < 1$ be any positive number.
\begin{itemize}
\item[i)] If $d(x,y)>\eta$ then $\|x-y\| > \eta\|x\|$ and $\|x-y\| > \eta\|y\|$.
\item[ii)] If $\|x-y\| > \eta \|x\|$ then $d(x,y) > \frac{\eta}{1+\eta}$.
\end{itemize}
\end{cor}
\begin{proof}
The first part follows from part ii) of Proposition \ref{thm.EquivalentOfDistances} by choosing $\ee = \eta$. The second part follows from part i) of Proposition \ref{thm.EquivalentOfDistances} for $\ee = \frac{\eta}{\eta+1}$, so that $\frac{\ee}{1-\ee} = \eta$.
\end{proof}

\begin{rem}
From now on, balls with respect to the projective distance will be called \emph{metric balls}, and balls with respect to the norm-induced distance $\|x-y\|$, will be called \emph{norm balls}.
\end{rem}

\begin{defn} \label{def.Cover}
Let $(X,d)$ be a semi-metric space. We denote the closed metric ball around $x_0 \in X$ of radius $r$ by $\B_{d,x_0,r}$. For a set $U\subset X$ and a number $r>0$, an $r$-metric cover of $U$ is a collection of points $\{x_i\}_{i=1}^n$ in $X$ such that $U \subseteq \bigcup_{i=1}^n \B_{d,x_i,r}$. The cover index $\N_d(U,r)$ is defined as the smallest possible size of an $r$-metric cover of $U$. If $X$ is a normed space, one similarly defines $r$-norm covers using norm balls $B_{\|\cdot\|,x_i,r}$, and the norm-cover index $\N_{\|\cdot\|}(U,r)$ as the smallest possible size of an $r$-norm cover of $U$.
\end{defn}

\section{Sample Complexity Bounds}
In this section, we prove that the sample complexity of giving $\ee$-close estimation of the $\LL_2$-gain is roughly equal to the sample complexity of giving $\ee/2$-close approximation of the system in the operator norm. We do so by showing that both sample complexities can be understood in terms of the metric-cover index $\N_d(U,\cdot)$ with an appropriately chosen radius. From now on, $d$ will denote the projective distance for $\X = \LL_2$. We note that by Remark \ref{rem.NaturalHarder}, it's enough to give a lower bound on the sample complexity of providing $\ee$-close estimation of the $\LL_2$-gain, and an upper bound on the sample complexity of providing $\ee/2$-close approximation of the system in the operator norm. We start with the former, showing the sample complexity is at least as big as some cover index:
\begin{thm}\label{thm.LowerBound}
Let $\U$ be any subset of $\LL_2(I)$ which does not contain $0$, let $L,\ee > 0$, and let $N>0$ be an integer. Let $\A$ be any $N$-sample $\LL_2$-gain estimation algorithm which provides $\ee$-close overestimation of the $\LL_2$-gain for any $H \in {\rm Lip}_L(\U,\LL_2)$. Then $N\ge\N_d(\U,\frac{\ee}{L})$. In particular, $\N_{\LL_2}(\U,\ee,L) \ge \N_d(\U,\frac{\ee}{L})$.
Moreover, if $\U$ has infinitely many elements, there does not exists an $N$-sample $\LL_2$-gain estimation algorithm which provides $\ee$-close overestimation of the $\LL_2$-gain for all Lipschitz operators $H:\U\to \LL_2$.
\end{thm}

\begin{proof}
We start by proving the first claim through contradiction. We assume, without loss of generality, that $N=\N_d(\U,\frac{\ee}{L})-1$, and want to show that $\A$ cannot provide $\ee$-close overestimation of the $\LL_2$-gain for any $H\in{\rm Lip}_L(\U,\LL_2)$. Let $H_0:\U\to\LL_2$ be the zero operator, defined by $H_0(u) = 0$ for all $u\in\U$. Run $\A$ on $H_0$, taking a total of $N$ samples from $\mathcal{U}$. We denote these samples by $u_1,\cdots,u_N$, let $y_i = H_0(u_i) = 0$ for $i=1,\cdots N$, and define a function $\rho:\U \to \R$ by $\rho(u) = \min_{i=1,\cdots,N} {\|u - u_i\|}$.
The function $\rho$ can be easily verified to be $1$-Lipschitz. Now, take any function $f\in \LL_2$ such that $\|f\|=1$, and let $H_1:\U \to \LL_2$ be defined as $H_1(u) = L\rho(u)f$. We first note that $H_1$ is $L$-Lipschitz. Indeed, for any $u,v\in \U$:
\begin{align*}
\|H_1(u)-H_1(v)\| = L|\rho(u)-\rho(v)|\|f\| \le L\cdot ||u-v||,
\end{align*}
where we use the fact that $\rho$ is 1-Lipschitz. Moreover, the definition of $\rho$ shows that $H_1(u_i) = 0 = y_i$ for $i=1,\cdots,N$. Thus, during the course of its run, $\A$ cannot differentiate between $H_0$ and $H_1$, and it issues the same estimate $\gamma$ for the $\LL_2$-gain of both. As the algorithm always issues an overestimate of the true $\LL_2$-gain, we have that $\gamma$ is no smaller than the $\LL_2$-gain of both $H_0$ and $H_1$, the former being equal to $0$. Thus, it's enough to prove that the $\LL_2$-gain of $H_1$ is bigger than $\ee$. 

By definition of the cover index, there exists some $u \in \U$ such that $d(u,u_i) > \frac{\ee}{L}$ for all $i=1,\cdots,N$. By Corollary \ref{cor.EquivDistances2}, we get that $\|u-u_i\| > \frac{\ee}{L}\|u\|$ for all $i=1,\cdots,N$, which implies that $\|H_1(u)\| = L\max_i \|u-u_i\| > \ee\|u\|$. In particular, the $\mathcal{L}_2$-gain of $H_1$ is bigger than $\ee$, so the output of the algorithm $\A$ must be bigger than $\ee$, and the error that $\A$ produces on $H_0$ is bigger than $\ee$. Thus, the algorithm $\A$ must take at least $\N_d(\U,\frac{\ee}{L})$ measurements in order to give an $\ee$-close overestimate of the $\mathcal{L}_2$-gain for any $L$-Lipschitz operator $\U\to \mathcal{L}_2(I)$, and $\N_{\mathcal{L}_2}(\ee,L,\U) \ge \N_d(\U,\frac{\ee}{L})$.

As for the second part of the theorem, it is enough to show that for any fixed $\ee>0$, $\N_d(\U,\frac{\ee}{L}) \to \infty$ as $L\to \infty$. Equivalently, we need to show that $\N_d(\U,\delta)\to\infty$ as $\delta \to 0$. As $\delta_1 < \delta_2$ implies that metric balls of size $\delta_1$ are smaller than metric balls of size $\delta_2$, we conclude that $\delta \mapsto \N_d(\U,\delta)$ is non-descending. Thus, it's enough to prove it is unbounded. Fix an arbitrary $N>0$, and we show that $\N_d(\U,\delta) \ge N$ for some $\delta > 0$. As $\U$ is infinite, we can find $N$ different points $v_1,\cdots,v_N$ in $\U$.  We now define $\eta = \min_{i,j} \frac{\|v_i-v_j\|}{\|v_i\|+\|v_j\|}$, and let $\delta < \frac{\eta}{\eta+1}$, so that $\frac{\delta}{1-\delta} < \eta$. We claim that no two points $v_i,v_j$ for $i\neq j$ can be within the same metric ball of radius $\delta$. Indeed,  suppose there exists some point $x$ such that $v_i,v_j$ are inside the metric ball of radius $\delta$ around $x$, where $i\neq j$. Then $d(x,v_i),d(x,v_j) \le \delta$, meaning that $\|x-v_i\| \le \frac{\delta}{1-\delta}\|v_i\|$ and  $\|x-v_j\| \le \frac{\delta}{1-\delta}\|v_j\|$. By the triangle inequality, we get $\|v_i-v_j\| \le \frac{\delta}{1-\delta}(\|v_i\|+\|v_j\|)$, or $\frac{\|v_i-v_j\|}{\|v_i\|+\|v_j\|} \le \frac{\delta}{1-\delta}<\eta$, which cannot hold by the definition of $\eta$. Thus, no two of the points $v_1,\cdots,v_N$ can lie in the same metric ball of radius $\delta$, hence $\N_d(\U,\delta) \ge N$. This completes the proof of the Theorem.
\end{proof}

After achieving a lower bound for the sample complexity of providing $\ee$-close estimation of the $\LL_2$-gain, we move to give an upper bound on the sample complexity of providing $\ee$-close approximation of the system in the operator norm:
\begin{thm}\label{thm.UpperBound}
Let $\U$ be any subset of $\LL_2$ which does not contain $0$, and assume $N=\N_d(\U,\frac{\ee}{2L+\ee})<\infty$. There exists an $N$-sample norm-approximation algorithm which provides $\ee$-close operator norm approximation for all $H\in {\rm Lip}_L(\U,\LL_2)$. In particular, $\N_{\rm op}(\U,\ee,L) \le\N_d(\U,\frac{\ee}{2L+\ee})$.
\end{thm}
\begin{proof}
Let $u_1,\cdots,u_N$ be a cover of $\U$ using metric balls of radius $\frac{\ee}{2L+\ee}$. Given the unknown $L$-Lipschitz operator $H$, we make measurements of the form $(u_i,y_i=H(u_i))$ for $i=1,\cdots,N$. Take an arbitrary operator $H_1$ in the set $
\left\{G:\U\to \mathcal{L}_2|\ G\text{ is $L$-Lipschitz }, G(u_i)=y_i, i=1,\cdots,N\right\}.$
The set is nonempty as it contains $H$. We claim that $\|H_1-H\|\le \ee$. Indeed, take any $u \in \U$, and we want to show that $\|H(u) - H_1(u)\| \le \ee\|u\|$. By Definition \ref{def.Cover}, there exists a point $u_i$ such that $d(u,u_i) < \eta=\frac{\ee}{2L+\ee}$. Proposition \ref{thm.EquivalentOfDistances} implies that $\|u-u_i\|\le\frac{\eta}{1-\eta}\|u\| = \frac{\ee}{2L}\|u\|$. By the triangle inequality, and $H,H_1$ being $L$-Lipschitz operators, we get:\
\begin{align*}
\|H(u)-H_1(u)\| &\le \|H(u) - H(u_i)\| + \|H(u_i) - H_1(u_i)\| +\\& ~~~~\|H_1(u_i) - H_1(u)\| \\&\le L\|u-u_i\|+\|y_i-y_i\|+L\|u-u_i\| \\&\le 2L\frac{\ee}{2L}\|u\| = \ee\|u\|.
\end{align*}\normalsize
Thus, the algorithm we suggested, sampling the inputs $u_1,\cdots,u_N$ and taking some $L$-Lipschitz operator consistent with the data, provides a solution to the norm approximation problem with error no more than $\ee$. 
\end{proof}
\begin{rem}
The algorithm solving the norm-approximation problem needs to take an arbitrary point from the set 
\begin{align*}
\left\{G:\U\to \mathcal{L}_2:\ G\text{ is $L$-Lipschitz }, G(u_i)=y_i, i=1,\cdots,N\right\}.
\end{align*}
We know that the set is nonempty, as the operator $H$ lies inside it. However, constructing a point within this set can be difficult. Generally, this part is equivalent to an Empirical Risk Minimization (ERM) step, which tries to find an $L$-Lipschitz operator $G$ which minimizes $\sum_{i=1}^n \|y_i - G(u_i)\|^2$. It is known that in some cases, solving the ERM problem can be computationally hard as the size of $\U$ increases \cite{ShalevShwartz2014}. However, in the case we present here, we can use Kirszbraun's theorem from nonlinear functional analysis \cite{Schwartz1969}, constructing $G$ as a piecewise linear function by linear interpolation \cite{Akopyan2008}, or even give an explicit formula for it \cite{Azagra2018}.
\end{rem}

Combining Theorem \ref{thm.LowerBound} and Theorem \ref{thm.UpperBound}, we get the following result on the sample complexity of learning the $\mathcal{L}_2$-gain of an $L$-Lipschitz operator, and of approximating it in the operator norm:
\begin{thm}
Let $\U \subseteq \mathcal{L}_2$ be a subset with infinitely many elements which does not contain $0$, and let $L > 0$ be any number. Then for any $\ee < L$, the following inequality holds:
\begin{align*}
\N_{\rm op} \left(\U,\frac{2\ee}{1-\ee/L},L\right) \le \N_{\mathcal{L}_2}(\U,\ee,L) \le \N_{\rm op} \left(\U,\frac{\ee}{2},L\right).
\end{align*}
In particular, for every $\ee < L/2$, we get:
\begin{align*}
\N_{\rm op} \left(\U,4\ee,L\right) \le \N_{\mathcal{L}_2}(\U,\ee,L) \le \N_{\rm op} \left(\U,\frac{\ee}{2},L\right). 
\end{align*}
\end{thm}
Informally, if $N$ samples suffice to give an $\ee$-close overestimation of the $\LL_2$-gain of a system, they suffice to give an $4\ee$-close approximation of the system in the operator norm.
\begin{proof}
The right half of the inequality follows from Remark \ref{rem.NaturalHarder}. As for the left half, for any $\eta > 0$, we have
$
\N_{\rm op}(\U,\eta,L) \le \N_d(\U,\frac{\eta}{2L+\eta})
$
and
$
\N_{\LL_2}(\U,\ee,L) \ge \N_d(\U,\frac{\ee}{L}).
$
We choose $\eta = \frac{2\ee}{1-\ee/L}$, so that $\frac{\ee}{L} = \frac{\eta}{2L+\eta}$ gives $\N_{\rm op}(\U,\eta,L) \le \N_{\LL_2}(\U,\ee,L)$, proving the first part of the corollary. As for the second part, we note that $\N_{\rm op} (\U,\eta,L)$ decreases as $\eta$ increases, and that $\frac{2\ee}{1-\ee/L} < 4\ee$ whenever $\ee < L/2$.
\end{proof}

\section{Estimating the Covering Index}
In the previous section, we showed the sample complexities $\N_{\rm op}(\U,\ee,L)$ and $\N_{\LL_2}(\U,\ee,L)$ are connected to the metric cover index $\N_d(\U,\cdot)$. We want to understand how these grow as the tolerance level $\ee$ decreases.  We do so by estimating $\N_d(\U,\eta)$ for $0<\eta<1$ using the norm cover index, as it has been extensively studied, mainly in the fields of learning theory and Gaussian processes \cite{ShalevShwartz2014,Dudley1967}. We first connect the metric cover index to the norm cover index:
\begin{prop} \label{prop.NormCoverEquivalence}
For any set $\U\subseteq \LL_2$ which does not contain $0$, and any $\eta\in(0,1)$,
\begin{align*}
\N_{\|\cdot\|}\left(\U,\frac{\eta}{1-\eta} \max_{u\in\U}\|u\|\right) \le \N_d(\U,\eta) \le \N_{\|\cdot\|}(\U,\eta\min_{u\in\U}\|u\|)
\end{align*}\normalsize
\end{prop}
\begin{proof}
Follows immediately from Proposition \ref{thm.EquivalentOfDistances}, which shows that any metric ball of radius $\eta$ contains a norm ball of radius at least $\eta\min_{u\in \U}\|u\|$, and is contained in a norm ball of radius at most $\frac{\eta}{1-\eta}\max_{u\in \U}\|u\|$.
\end{proof}
Thus, it's enough to estimate the cover index $\N_{\|\cdot\|}(\U,\cdot)$:

\begin{prop} \label{thm.CoverAsymptotics}
Suppose that $\U$ is a compact subset of a finite dimensional subspace $F$ of $\LL_2$, such that $\U \subseteq F$ has a nonempty interior. There exist constants $D_1,D_2>0$, depending only on $\U$, such that for any $\delta > 0$, 
\begin{align} \label{eq.CoverAsymptotics}
D_1 \delta^{-\dim(F)} \le \N_{\|\cdot\|}(\U,\delta) \le D_2 \delta^{-\dim(F)}
\end{align}
\end{prop}\normalsize

\begin{proof}
We start with the upper bound, for which we may assume without loss of generality that $\delta^{-1}$ is a positive integer. We let $n=\dim(F)$ and $\rho = \sup_{u\in \U} \|u\|$, and choose an orthonormal basis $f_1,\cdots,f_n$ to $F$.
We define the set $C$ as the collection of all points in $F$ of the form $\sum_{i=1}^n c_i f_i$, where the constants $c_i$ are taken in the set $\{\frac{m}{n}:\ m\in\{-\lceil\rho\rceil,-\lceil\rho\rceil+\delta,-\lceil\rho\rceil+2\delta,\cdots,\lceil\rho\rceil\}\}$. We claim that the set $C$ constitutes a $\delta$-norm cover for $\U$. Indeed, given a point $u\in \U$, we can write $u=\sum_{i=1}^n a_i f_i$ when $\sum_{i=1}^n a_i^2 = \|u\|\le \rho$, and in particular $\max_i|a_i| \le \rho$. We can thus find some $c=\sum_{i=1}^n c_i f_i\in C$ such that $\max_i |a_i-c_i| \le \delta$, and $\|u-c\| \le \delta\sqrt{n}$ by Parseval's equality. In particular, $C$ is a $\delta\sqrt{n}$-norm cover of $\U$, and it has $\left(\frac{\lceil\rho\rceil}{\delta}\right)^n$ points. We thus get that:
$
\N_{\|\cdot\|}(\U,\delta\sqrt{n}) \le \left(\frac{\lceil\rho\rceil}{\delta}\right)^n.
$
By replacing $\delta$ with $\delta/\sqrt{n}$, we get:
\begin{align*}
\N_{\|\cdot\|}(\U,\delta) \le \left(\frac{\lceil\rho\rceil\sqrt{n}}{\delta}\right)^n = D_2\left(\frac{1}{\delta}\right)^n,
\end{align*}
where $D_2 = (\sqrt{n}\lceil\rho\rceil)^n>0$ is a constant depending on $\U$.

We now move to the left-hand side of the inequality \eqref{eq.CoverAsymptotics}. Consider the invertible linear map $S:F\to \R^n$ defined by $F(\sum_{i=1}^n a_i f_i) = (a_i)_{i=1}^n$. By definition, given any $\delta > 0$, we can find points $\{x_i\}_{i=1}^{\N_{\|\cdot\|}(\U,\delta)}$ such that
$
\U \subseteq \bigcup_{i=1}^{\N_{\|\cdot\|}\U,\delta)} \B_{\|\cdot\|,x_i,\delta},
$
where we recall that $\B_{\|\cdot\|,x_i,\delta}$ is the norm ball around $x_i$ of radius $\delta$. By applying $S$, we conclude that
$
S(\U) \subseteq \bigcup_{i=1}^{\N_{\|\cdot\|}(\U,\delta)} B_{\R^n,z_i,\delta},
$
where $z_i = S(x_i)$ and $B_{\R^n,z_i,\delta}$ is the ball around $z_i \in \R^n$ of radius $\delta$ with respect to the Euclidean distance. Let $\vol(\cdot)$ be the volume in $\R^n$, as computed using the Lebesgue measure on $\R^n$. Then:
\begin{small}
\begin{align*}
\vol(S(\U)) \le \vol\left(\bigcup_{i=1}^{\N_{\|\cdot\|}(\U,\delta)} B_{\R^n,z_i,\delta}\right)\le \sum_{i=1}^{\N_{\|\cdot\|}(\U,\delta)}\vol\left(B_{\R^n,z_i,\delta}\right),
\end{align*}
\end{small}
which is equal to  $\N_{\|\cdot\|}(\U,\delta)V_n\delta^n$, where $V_n$ is the volume of the unit ball in $\R^n=\R^{\dim(F)}$. We thus get that:
$
\N_{\|\cdot\|}(\U,\delta) \ge \frac{\vol(S(\U))}{V_n\delta^n} = D_1 \left(\frac{1}{\delta}\right)^n,
$
where $D_1 = \vol(S(\U)) / V_n$ is a constant which depends only on $\U$.
\end{proof}

Proposition \ref{thm.CoverAsymptotics} gives a two-sided estimate on the norm-cover index of $\U$, which is in turn related to the sample complexities $\N_{\LL_2}(\U,\ee,L)$ and $\N_{\rm op}(\U,\ee,L)$. We prove the following theorem, showing the sample complexities are asymptotically equivalent, which is the main result of this paper.

\begin{thm} \label{cor.Asymptotics}
Suppose that $\U$ is a compact subset of a finite dimensional subspace $F$ of $\LL_2$, such that $\U\subseteq \LL_2$ has a nonempty interior.  For any $c<1$, there exists constants $C_1,C_2>0$, depending only on $\U$ and $c$, such that for any $\ee,L > 0$ with $\ee < cL$,
\begin{align*} 
C_1 \N_{\rm op} (\U,\ee,L) \le \N_{\LL_2} (\U,\ee,L) \le C_2 \N_{\rm op} (\U,\ee,L)
\end{align*}
\end{thm}

\begin{proof}
We denote $n = {\rm dim}(F)$. Propositions \ref{prop.NormCoverEquivalence} and \ref{thm.CoverAsymptotics}, Remark \ref{rem.NaturalHarder}, and Theorem \ref{thm.UpperBound} show that:\small
\begin{align*}
&\N_{\LL_2}(\U,\ee,L) \le \N_{\rm op}\left(\U,\frac{\ee}{2}\right) \le \N_{d}\left(\U,\frac{\ee/2}{2L+\ee/2}\right)\le \\ &\N_{\|\cdot\|}\left(\U,\frac{m\ee/2}{2L+\ee/2}\right)  \le D_2 \left(4+\frac{\ee}{L}\right)^n \left(\frac{L}{m\ee}\right)^n \le K_2 \left(\frac{L}{\ee}\right)^n
\end{align*}\normalsize
where $D_2>0$ is the constant in Proposition \ref{thm.CoverAsymptotics}, $m = \min_{u\in \U}\|u\|$, and $K_2 = D_2 \left(\frac{4+c}{m}\right)^n > 0$ is a constant depending only on $\U$ and $c$. Similarly, one can use Propositions \ref{prop.NormCoverEquivalence} and \ref{thm.CoverAsymptotics}, and Theorem \ref{thm.LowerBound} to prove that $\N_{\LL_2}(\U,\ee,L) \ge K_1 \left(L/\ee\right)^{n}$ when $K_1$ is a constant depending only on $\U$ and $c$. Thus, Remark \ref{rem.NaturalHarder} implies that:
\begin{align*}
K_1\left(\frac{L}{\ee}\right)^{n} \le \N_{\LL_2}(\U,\ee,L) \le \N_{\rm op}\left(\U,\frac{\ee}{2},L\right) \le 2^{n}K_2\left(\frac{L}{\ee}\right)^{n}.
\end{align*}
Taking $C_1 = \frac{K_1}{K_2}$ and $C_2 = 4^n\frac{K_1}{K_2}$ completes the proof.
\end{proof}

\begin{rem}
Proposition \ref{thm.CoverAsymptotics} assumes the set $\U$ is contained in some finite-dimensional subspace of $\LL_2$. One can ask what happens when $\mathcal{U}$ is not contained in a finite-dimensional subspace. Given $\delta > 0$, we will say that $\mathcal{U}$ is $\delta$-finite dimensional if there exists a finite dimensional subspace $F_\delta$ such that $\U$ is contained in an $\delta$-neighborhood of $F_\delta$. It is straightforward to prove that if $\U$ is covered by $\B_{\|\cdot\|,x_1,\delta},\cdots,\B_{\|\cdot\|,x_N,\delta}$, then it is contained in an $\delta$-neighborhood of the finite-dimensional space $F = {\rm span}\{x_1,\cdots,x_N\}$. Thus, if $\N_{\|\cdot\|}(\U,\delta)$ is finite then $\U$ is $\delta$-finite dimensional. In that case, one can prove Proposition \ref{thm.CoverAsymptotics} still holds, where $F$ is replaced with $F_\delta$. Otherwise, $\N_{\|\cdot\|}(\U,\delta) = \infty$, so no algorithm can solve the $\LL_2$-estimation nor the norm-approximation problems with high precision (due to Propositions \ref{prop.NormCoverEquivalence} and \ref{thm.CoverAsymptotics} and Theorems \ref{thm.LowerBound}, \ref{thm.UpperBound}).
\end{rem}

\section{Conclusions and Outlook}
In this paper, we study the number of samples needed to give an overestimate of the $\LL_2$-gain of some unknown operator $H$, or to give an approximation to $H$ in the operator norm. These are known as the sample complexities for the said problems. We used the notion of cover indices, and namely the cover index for the projective distance, to give a bound on the sample complexity of overestimating the $\LL_2$-gain of the operator $H$, and of giving an approximation for $H$ in the operator norm. We then studied the asymptotics of the cover index for the projective distance using the cover index for the norm-induced metric, which resulted in a bound on the studied sample complexities, showing that the number of samples needed to give an $\ee$-close upper bound on the $\LL_2$-gain of the operator $H$ is roughly the same as the number of samples needed to give an $\ee$-close approximation for the operator $H$ in the operator norm, which intuitively should required more data. We do not claim that approximation of the nonlinear operator $H$ directly gives rise to a control law, as it will usually be too complex to work with directly, unless one uses a small-gain approach. However, the results show that even though data-driven methods need not learn a model for the system, they cannot avoid using at least the same amount of data. Future research on this problem can try and extend these results in three different ways - switch the $\LL_2$-gain with shortage of passivity and cone constraints, change the sampling model to include noise and disturbances, or to a one in which we cannot choose the input (corresponding to a machine-learning scheme), or consider specific systems instead of general Lipschitz operators, e.g. bi-linear, polynomial, or trigonometric systems.
\vspace{-5pt}
\bibliographystyle{ieeetr}
\bibliography{main}

\end{document}